\documentclass[leqno,onefignum,onetabnum]{siamltex1213}
\usepackage{amsmath}
\usepackage{amssymb}

\title{A Study on Splay Trees\thanks{ This work was supported by
    national funds through Funda\c{c}\~{a}o para a Ci\^{e}ncia e a
    Tecnologia (FCT) with reference UID/CEC/50021/2013 and DataStorm
    EXCL/EEI-ESS/0257/2012 and European Union's Horizon 2020 research and
innovation programme under the Marie Sk{\l}odowska-Curie Actions grant
agreement No 690941.}}

\author{Lu\'{i}s~M.~S.~Russo\thanks{INESC-ID and the Department of
    Computer Science and Engineering, Instituto Superior T\'{e}cnico,
    Universidade de Lisboa (\email{luis.russo@tecnico.ulisboa.pt}).}}

\begin{document}
\maketitle
\slugger{sicomp}{xxxx}{xx}{x}{x--x}

\begin{abstract}
%
  We study the dynamic optimality conjecture, which predicts that
  splay trees are a form of universally efficient binary search tree, for any
  access sequence. We reduce this claim to a regular access bound,
  which seems plausible and might be easier to prove.
  This approach may be useful to establish dynamic
  optimality\footnote{This draft version,
as it exists immediately prior to editing and production by the
Publisher will be posted on the noncommercial pre-print server
arXiv.org, in accordance with SIAM's Consent to Publish.}.
\end{abstract}



\pagestyle{myheadings}
\thispagestyle{plain}
\markboth{TEX PRODUCTION}{USING SIAM'S \LaTeX\ MACROS}

\section{Introduction} 
\label{sec:introduction}
\hfill

Binary search trees (BSTs) are ubiquitous in computer science. Their
relevance is well established, both in theory and in
practise. Figure~\ref{fig:intro} illustrates this data structure. The
numbers in the nodes represent the keys. If read left to right, the
keys form an ordered sequence, shown bellow the tree. This
property can be used to efficiently determine if a given number exists
in the tree. If we want to check if the tree contains the number
$0.55$ we can start by comparing this value with the one at the
root, $0.56$. Since the number we are searching for is larger than
$0.56$, the search continues in the left sub-tree, i.e., it 
proceeds to the node containing $0.40$, thus discarding all the
sub-tree to the right. It is the process of eliminating a large
portion of its search space that makes BSTs efficient.
\begin{figure}[h]
  \centering
  \input{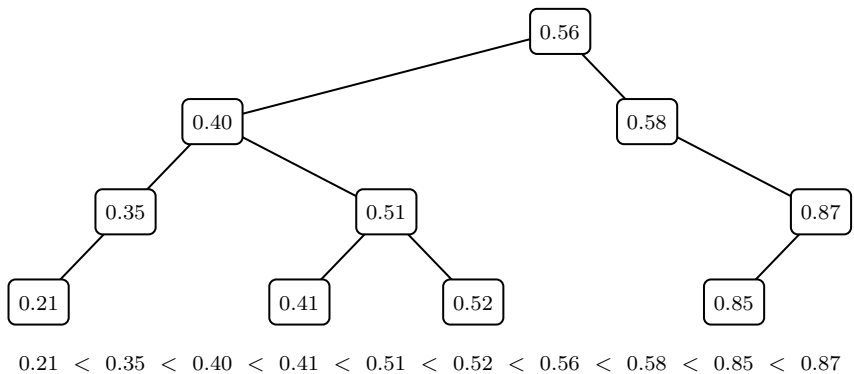}
  \caption{Binary search tree containing keys from $[0,1]$.}
  \label{fig:intro}
\end{figure}

As the search goes through the tree it will eventually reach the
node containing $0.52$, at which point it becomes apparent that
 $0.55$ is not present in the structure. In fact this search
 concludes that no number in the open interval $]0.52, 0.56[$
exists in the tree. This is the only valid conclusion. Even
though the search passed through nodes $0.40$ and $0.51$ it would not
be valid to conclude that no number in $]0.40,0.51[$ exists, because
$0.41$ is on the tree. Hence whenever a search fails we still obtain
information from the data structure. In particular we obtained the
numbers in the tree that are the successor ($0.56$) and the
predecessor ($0.52$) of $0.55$.

Consider how the shape of the tree affects the performance of the
queries. If we are computing a single, isolated, query, we
may want to guarantee that, at each node, the size of the left and
right sub-trees is roughly the same. Note that we never know which
sub-tree a given search will choose. If both sub-trees have the same
size at least half the search space is discarded at each
step. Alternatively we could impose that the tree does not contain
long branches, since a search might have to traverse all of it.

Either way, some policy must be used to main a tree shape that reduces
the search time. Determining ``good'' tree shapes is a complex task,
specially when we consider a sequence of queries, instead of a single
query. Moreover, must of the time, this sequence is not known a
priori.  If a sequence of a couple million queries searches for $0.21$
in half of them and only once for $0.85$, then it might be better to
keep $0.21$ at the root and leave $0.85$ on the longest existing
branch. We would expect such a shape to reduce the overall
time. Hence a balanced tree is important if we need to answer several
essentially distinct queries. On the other hand if the queries are
strongly biased towards a specific region of the keys maintaining
certain branches shorter than the rest might be necessary.

Several strategies are known to maintain efficient BST shapes, see
Knuth~\cite{Knuth:1998:ACP:280635} for an introduction to the
subject. In this paper we focus on the approach used by splay
trees~\cite{Sleator:1985:SBS:3828.3835}. Our goal is to show that this
approach is optimal, in the sense that no other strategy can be
asymptotically faster than splay trees. Notice that splay trees
dynamically alter their structure as they process queries. Hence we
assume that any other BST can do the same.  Figure~\ref{fig:introSp}
illustrates the tree that results from accessing $0.21$ on a splay
tree.
\begin{figure}[tb]
  \centering
  \input{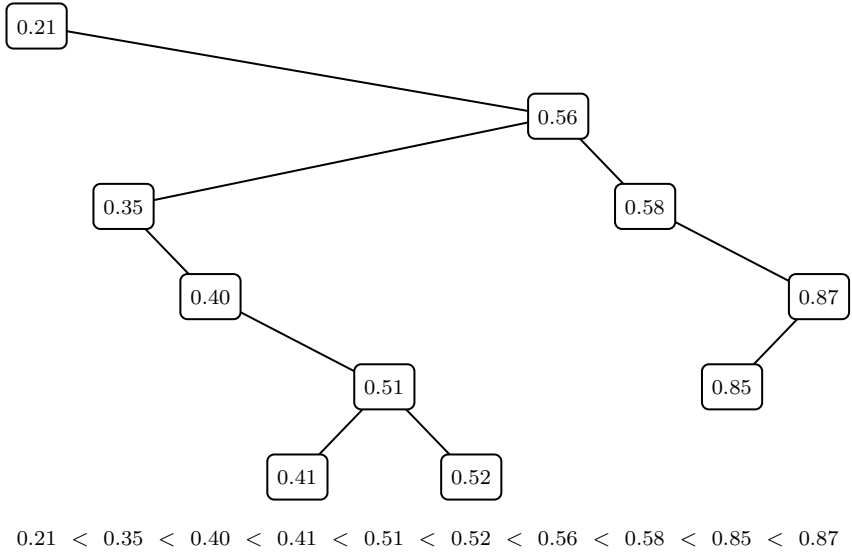}
  \caption{After splaying $0.21$.}
  \label{fig:introSp}
\end{figure}

Our main contribution is the following:
\begin{itemize}
\item We, almost, show that splay trees are dynamically optimal. This
  means that no other BST is asymptotically faster than splay trees,
  no matter what is the structure of the query sequence or the
  re-shaping policy it uses. This property was
  conjectured to be true, over 30 years
  ago~\cite{Sleator:1985:SBS:3828.3835}.

  Our proof depends on a regular access property, which we deem
  plausible and conjecture to be true, see
  Section~\ref{sec:analysis-overview}. At this time we do not have a
  proof of this property.
\end{itemize}
\section{The Problem}  
\label{sec:problem}
\hfill

The introduction exemplified how a BST can be used to maintain a
finite set $K$ of, keys, elements of a universe set $X$. In our
example $X$ is $\mathbb{Q}$. Our goal is to maintain a data structure
such that given a fraction $x$ from $\mathbb{Q}$ we can determine if
$x$ is in $K$, or not. We refer to these operations as queries. A
query is successful it $x$ does belong to $K$ and unsuccessful
otherwise.

Several efficient data structures exist for this problem, depending on
which resources are critical and on which extra operations are
necessary. For our example we also want map behaviour. This means that
the elements of $K$ contain information. Whenever $x$ does exist in
$K$ we also want to access the associated information. The two major
classes of data structures that can be used in this scenario are trees
and hashes. The set $\mathbb{Q}$ is, in general, referred to as the key
universe and the associated information is referred  to as the value
set. In this paper we omit this latter set.

Contrary to hashes, BSTs rely on the order relation among the elements
of the universe. In $\mathbb{Q}$ we have for example that $(1/4) <
(1/2)$. Hence BSTs can efficiently determine successors and
predecessors. The successor of $x$ in $K$ is $\min\{y \in C | x <
y\}$. To simplify the analysis we do not consider unsuccessful
queries, i.e., we assume that all the queries find elements in $K$. An
unsuccessful query for $x$ can be modelled by two successful ones,
one for the successor and one for the predecessor. In general BSTs
support inserting and removing elements from $K$, but we also
do not study those operations.

The re-structuring policy of splay trees consists in moving the
accessed nodes, and the nodes in the respective path, upwards towards
the root. The precise operations are shown in
Figures~\ref{F:zig},~\ref{F:zigzig}~and~\ref{F:zigzag}, where the node
containing $x$ is the one being accessed. The configuration before the
access is represented on the left and the structure after the access
on the right. Accesses when $x$ is on the right sub-tree are obtained by
symmetry.
\begin{figure}[tb]
\begin{center}
\input{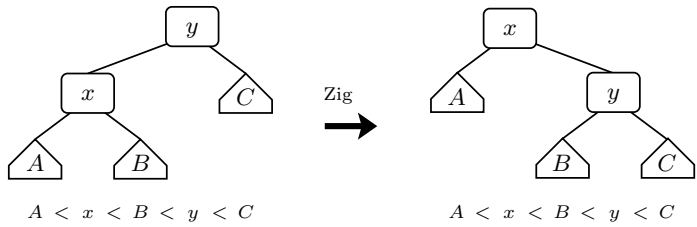}
\caption{Zig operation.}
\label{F:zig}
\end{center}
\end{figure}
\begin{figure}[htbp]
\begin{center}
\input{zigzig-R.gte}
\caption{Zig Zig operation.}
\label{F:zigzig}
\end{center}
\end{figure}
\begin{figure}[htbp]
\begin{center}
\input{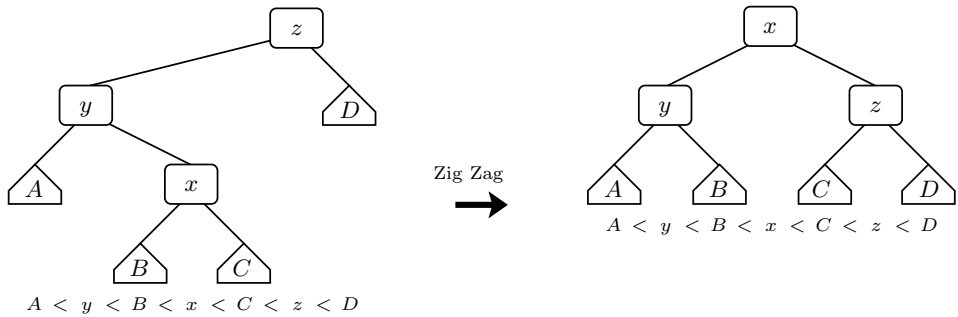}
\caption{Zig Zag operation.}
\label{F:zigzag}
\end{center}
\end{figure}

To compare splay trees with any other BST we assume that the other BST
works in the following way. Besides the tree itself there is a cursor
that moves between nodes. An algorithm on a BST may perform any of
the following operations:
\begin{description}
\item[Compare] the key at the cursor with the current
  search value.
\item[Move] the cursor to an adjacent node, left child, right child
  or parent.
\item[Rotate] the node upwards. The Zig operation is a rotation of
  node $x$, (Figure~\ref{F:zig}).
\end{description}
To perform a search in a BST the cursor starts at the root and
performs a sequence of the previous operations. The result of the
search is the node for which a comparison with $x$ was equal. This
node may, or may not, be the last one in the sequence. At the end of a
search the cursor must return to the root, those moves are accounted
for.

In this model we allow for constant extra information at each node,
such as colour for red-black trees or height for AVL trees. The
intention is to forbid dynamic memory blocks, which
could be used to implement hashing. Hence only a fixed constant amount
of information can be added to a node. Other BST models exist, that
 can be shown to be within a constant factor of the one we
consider~\cite{DBLP:journals/siamcomp/Wilber89}.

Among all BSTs we focus on the one that achieves
the best performance for a query sequence. This means choosing the
tree after knowing the complete query sequence, in other words,
offline. We refer to this optimal tree as $T$.

We do not count compare operations and therefore assume that they cost
0. For the optimal tree we count the number of moves and rotations,
but for the splay tree we only count the number of moves. This avoids
having to carry around a factor of 2 or 3, which would result from
counting rotations and or comparisons.

\section{Analysis Overview}     
\label{sec:analysis-overview}
\hfill

In this section we overview the main techniques in the analysis and
describe our new potential function.
\begin{description}
\item[Amortized analysis.] The analysis is amortized, meaning that the
  time a splay operation takes is not accounted per si, but
  considering preceding and succeeding
  operations~\cite{doi:10.1137/0606031}. Therefore the amount of time
  an operation requires can be partially shifted to some other
  operation in the sequence. The resulting time is known as the
  amortized time. The most common way to transform the total sum into
  an equivalent telescopic sum is to associate a ``potential'' value
  to each state of the data structure $D$, the value $\Phi(D)$. The
  amortized cost of an operation is then defined as $\hat{c} = c +
  \Phi(D') - \Phi(D)$, were $c$ is the actual cost and $\Phi(D')$ the
  potential after the operation is performed. Summing the previous
  equations over all $m$ queries, and using the fact that it is a
  telescopic sum, we obtain a global relation $\sum_{1\leq i \leq m}
  c_i = \left( \sum_{1\leq i \leq m} \hat{c}_i \right) -
  \Phi(D_m)$. In this equation $\Phi(D_m)$ represents the potential
  after the sequence, which might be a negative number. The value
  $\Phi(D_0)$ represents the potential at the beginning and it was
  omitted because it is assumed to be $0$.  This is a classical tool
  in the analysis of splay trees. We present a new potential
  function. This function is chosen so that the structure of the
  optimal tree gets transferred into the splay tree.
\item[Restricted Accesses.] After bounding the amortized cost of
  accesses in $S$ we proceed to amortize the cost of rotations in
  $T$. Such rotations alter $\Phi(D)$ and therefore most be accounted
  for. To obtain a constant bound on this variation we introduce
  extra, organizing splays, in $S$. The amortized time to compute
  these extra splays will depend on the depth of the node we are
  rotating in $T$. Therefore if $T$ performs two rotations in a row we
  need to count the node depth twice. To obtain a valid bound we 
  force $T$ to move its cursor back to the root after performing a
  rotation. We also further impose that $T$ may not access nodes of
  depth $3$ or more. Although these restrictions seem harsh we show
  that, with a linear slowdown factor, it is possible to force the
  optimal sequence to respect them, Lemma~\ref{lem:simulateT}.
\item[Organizing Splays.] Splay trees do not know the full sequence of
  queries in advance, fortunately the analysis does. This means that
  in the analysis it is possible to verify if the structure of the
  splay tree is adequate for the upcoming queries. We optimize the
  tree structure by adding extra organizing splays. These splays
  introduce the logic dependency on a regular access.
\end{description}
The regular access statement is the following:
\newtheorem{conjecture}[theorem]{Conjecture}
\begin{conjecture}[Regular Access]
\label{C:regular}
Let $c_1 + \ldots + c_{m}$ be the total cost of a sequence of $m$
splay operations and $c'_1 + \ldots + c'_{m+e}$ the cost of performing
the same sequence of splay operations, augmented with $e$ extra splays
anywhere. Then $c_1 + \ldots + c_{m} = O(c'_1 + \ldots + c'_{m+e})$.
\end{conjecture}

This conjecture basically states that computing more splays takes
longer. We assumed this property in our reasoning, until it became
evident that it was not trivial to establish. Our proof relies
precisely on this property because the extra splays are used to
optimize the tree for future accesses.

We will now define the potential function $\Phi$ and give an overview
of the results that we prove in Section~\ref{sec:details}. $T$ denotes
the optimal tree. To every node $v$ of $T$ we assign a weight $w_T(v)
= 1/4^{d_T(v)}$, where $d_T(v)$ is the depth of $v$, i.e., the
distance to the root. The root itself has depth $0$, its children have
depth $1$ and so on. For each node we add up all the weights of its
sub-tree, including $w_T(v)$ itself, the resulting sum is denoted
$s_T(v)$. The rank $r_T(v)$ is computed as $\log(s_T(v))$. The
potential of tree $T$ is given by the following sum:
\[
P(T) = \sum_{v \in T} r_T(v)
\]

Let us consider the trees in Figure~\ref{F:examplePhi}. We have that
$w_T(a) = w_T(c) = 1/16$, $w_T(b) = w_T(e) = 1/4$ and $w_T(d) =1$.
The resulting $r$ values are $r_T(a) = r_T(c) =\log(1/4^2) = -4$,
$r_T(e) = -2$, $r_T(b) = \log (1/4+1/8) = \log (3/8)$ and $r_T(d) =
\log (1 + 1/2 + 1/8) = \log (13/8)$. Hence the total potential of $T$ is
$P(T) = r_T(b) + r_T(d) -10$.

We use $S$ to denote the splay tree. There is a one to one relation
between the nodes of these trees, because both trees share the same
key values. For any node $v$, of $T$, its correspondent in $S$ is
$f(v)$, by correspondent we mean that $v$ and $f(v)$ store the same
key value.  This correspondence is used to transfer the weight values
from $T$ to $S$, more precisely set $w(f(v)) = w_T(v)$. We omit the
$S$ subscript to simplify notation. The values of $w$, $s$ and $r$
that do not have a subscript refer to $S$. Moreover we also omit the
function $f$ and assume instead that $v$ in $S$ means $f(v)$. Hence
the previous relation will be written as $w(v) = w_T(v)$. The values
$s(v)$ and $P(S)$ are computed as before. In general $s(v)$ and
$s_T(v)$ are not, necessarily, equal. See
Figure~\ref{F:examplePhi}. Likewise $P(T)$ and $P(S)$ are also not,
necessarily, equal. In fact we use $\Phi(D) = P(S) - P(T)$ as our
potential function.

 In our example we have that $r(a) = r_T(a)$, $r(c) =
r_T(c)$, $r(e) = r_T(e)$, $r(b)= r_T(d)$ and $r(d) = \log(1 + 1/4 +
1/16) = \log (21/16)$. Therefore 
$\Phi(D) = \log(21/16) - \log(3/8) = \log(7/2)$.
\begin{figure}[htbp]
\begin{center}
\input{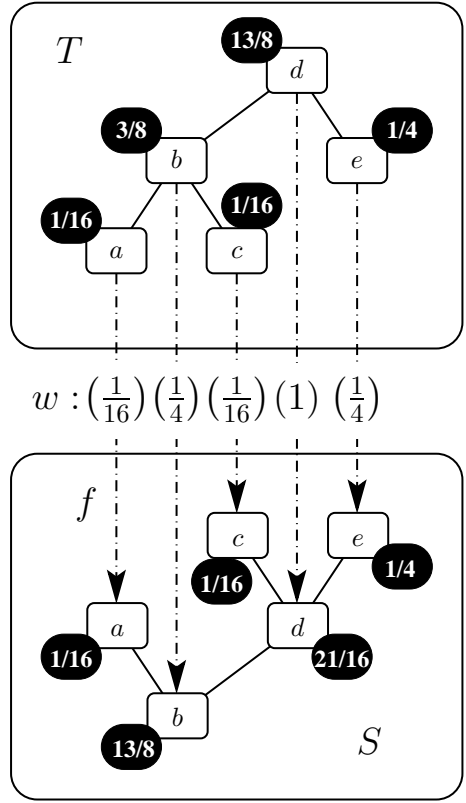}
\caption{Optimal tree $T$ on top, splay tree $T$ in the bottom, upside down.
Example of computing $\Phi$. The $f$ function is illustrated
  by the dashed arrows.  The weight values $w$ are show in between the
  two trees. Inside the black rectangles we show the $s$ values
  associated with the nodes. }
\label{F:examplePhi}
\end{center}
\end{figure}

In the analysis we assume that $S$ and $T$ take turns. First $S$
searches for $x$ and $T$ remains idle. Then $T$ searches for the same
$x$, while $S$ remains idle. Hence whenever the structure of $S$ changes
$P(S)$ changes, but $P(T)$ remains constant. However when the
structure of $T$ changes, both $P(T)$ and $P(S)$ change.

The structure of our argument is as follows:

\begin{itemize}
\item We show that $-n < \Phi(D_m)$, Lemma~\ref{lem:bound-phi}, where
  $n$ is the number of keys.
\item We show that the amortized cost
to splay a node $v$, in $S$, is at most $O(1 + d_T(v))$,
Lemma~\ref{sec:amortized-costs-s}.
\item  We show that whenever there is a
rotation at node $v$ in $T$ we have $\Delta \Phi = O(1)$, by splaying
at most 3 nodes, whose depth in $T$ is at most $d_T(v)$,
Lemma~\ref{lem:powerDelta}.
\item  We combine these results with a restricted
optimal sequence, Lemma~\ref{lem:simulateT}, and obtain an
optimality result, Theorem~\ref{teo:optimal}.
\end{itemize}
Thus we establish that if the optimal tree, with $n$
nodes, needs $R$ rotations and $M$ moves to process a given query
sequence then splay trees require $O(n + R + M)$ time, provided they
have the regular access property.
\section{The Details}  
\label{sec:details}
\hfill

\begin{figure}
  \centering
  \begin{tabular}{ll}
    $T$ & The optimal binary search tree.  \\
    $T'$ & A restricted version of $T$. \\
    $S$ & The splay tree. \\
    $n$ & Number of nodes in $T$ and $S$ \\
    $m$ & Number of queries in a sequence. \\
    $e$ & Number of added organizing splays. \\
    $R$ & Number of rotations in $T$ \\
    $M$ & Number of cursor moves in $T$ \\
    $R'$, $M'$ & Same as $R$ and $M$, but in $T'$ \\
    $v$, $u$ & Nodes \\
    $t$ & The root node \\
    $d_T(v)$ & The depth of node $v$ in $T$ \\
    $w(v)$ & Weight of node $v$ \\
    $s(v)$ & Sum of the weights, in the sub-tree bellow $v$ \\
    $r(v) = \log(s(v))$ & Rank for node $v$ \\
    $w_T(v)$, $s_T(v)$, $r_T(v)$ & Weight, sum and rank in $T$ \\
    $\Phi(D)$ & Potential value for a given state of $S$ and $T$
  \end{tabular}
  \caption{Symbol table.}
  \label{fig:notation}
\end{figure}
In this section we fill in the details that are summarized in the
previous section. This section is divided into three parts. We start
by studying the properties of our potential function $\Phi$. We then
shift the focus to $T$ and explain how to restrict the general access
sequence to ensure that it respects certain desirable conditions. The
last part presents bounds for the amortized costs of accessing elements
in $S$ and rotating nodes in $T$. We then combine these results to
obtain an optimality result.
\subsection{Properties of $\Phi$}
\label{sec:properties-phi}
\hfill

Let us start by analyzing the potential function $\Phi(D)$. We assume
that the initial configuration of $S$ and $T$ is the same. Therefore
 $P(S)=P(T)$ and  $\Phi(D_0) = 0$.

Given that we are using base $2$ logarithms it would be natural to use
base $(1/2)$ powers for the weights $w$ and $w_T$. Using $(1/4)$
instead causes the amortized access cost to be twice as big. On the
other hand we obtain some handy properties.
\begin{lemma}
\label{lem:boundSV}
Let $v$ be a node. The following bounds hold:
  \begin{align}
    0 &\leq w(v) &&\mbox{ for any $v$ in $T$ or $S$.} \label{eq:basic0} \\
    w(v) &\leq 1 &&\mbox{ for any $v$ in $T$ or $S$.} \label{eq:basic4} \\
    w(v) &\leq s(v) &&\mbox{ for any $v$ in $T$ or $S$.} \label{eq:basic1} \\
    s_T(v) &< 2 \times w_T(v)&& \mbox{ only for $v$ in  $T$.} \label{eq:basic2} \\
    s(v) &< 2 &&\mbox{ for any $v$ in $T$ or $S$ } \label{eq:basic3} \\
    s(t') &= s_T(t) &&\mbox{ for the root $t'$ of $S$ and $t$ of $T$ } \label{eq:basic5}
  \end{align}
\end{lemma}
\begin{proof}
  Inequalities~(\ref{eq:basic0}) and~(\ref{eq:basic4}), follow
  directly from our definition of weights.

Inequality~(\ref{eq:basic1}), follows from the definition of
$s$ and Inequality~(\ref{eq:basic0}).

Inequality~(\ref{eq:basic2}) use a geometric series. For the nodes $v$
in $T$ the value $s_T(v)$ is a sum of powers of $1/4$, which depends
on the topology of tree. If $T$ is a single branch then
$s_T(v) < w_T(v) (1 + (1/4) + (1/4^{2}) + \ldots + (1/4^j) + \ldots) < (4
w_T(v)/3)$. If $T$ is perfectly balanced,
this bound becomes even larger. In that case $s_T(v) < w_T(v)(1 +
2 \times (1/4) + 2^2 \times (1/4^2) + \ldots + 2^j \times (1/4^j) +
\ldots) < 2
\times w_T(v)$. In this case the bound comes from the geometric series
of $(1/2)$. Moreover the remaining cases will also be bounded by the
geometric series of $(1/2)$.

Inequality~(\ref{eq:basic3}) holds for $T$, because of
inequalities~(\ref{eq:basic4}) and~(\ref{eq:basic2}). For $S$ note
that in general $T$ cannot have more than $2^d$ nodes with depth $d$
and consequently $S$ cannot have more that $2^d$ nodes with weight
$(1/4^d)$, hence we obtain the geometric series of $(1/2)$ again.

For equality~(\ref{eq:basic5}) notice that the sums at the root
contain all the weight values. Since these values are
mapped from $T$ to $S$ they are globally the same and therefore so is
their sum.
\end{proof}

In particular from Inequality~(\ref{eq:basic3}) we deduce that $r(t)
\leq 1$, for the root $t$ of $S$.

Using potential functions that can assume negative values implies that
the value of $\Phi(D_m)$ becomes a term in the total time. Hence we
bound this value.
\begin{lemma}
\label{lem:bound-phi}
The bound $-n < \Phi(D)$ holds for any configuration of $S$ and $T$,
where $n$ is the number of keys.
\end{lemma}
\begin{proof}
The following derivation establishes the result.
\begin{eqnarray}
  P(S) - P(T) & = & \Phi(D) \\
  \sum_{v \in S} \log(w(v)) - P(T) & \leq & P(S) - P(T)\label{eq:2}\\
  \sum_{v \in S} \log(w(v)) - \sum_{v \in T} \log(2w_T(v)) & < & \sum_{v \in S}
  \log(w(v)) - P(T) \label{eq:3}\\
  -\sum_{v \in T} 1 & = & \sum_{v \in S} \log(w(v)) - \sum_{v \in T} \log(2w_T(v)) \label{eq:4} \\
  -n & = &  -\sum_{v \in T} 1
\end{eqnarray}
To obtain Inequality~(\ref{eq:2}) use Inequality~(\ref{eq:basic1}), for
all the terms of the sum in $P(S)$.
For Inequality~(\ref{eq:3}) use Inequality~(\ref{eq:basic2}) and
apply logarithms and change signals.
Equation~(\ref{eq:4}) follows from the same argument as
Equation~(\ref{eq:basic5}) and the fact that $\log(2w_T(v)) = 1 + \log(w_T(v))$.
\end{proof}

\subsection{Restricting $T$}
\label{sec:restricting-t}
\hfill

Competing directly with $T$ is fairly hard, partially because the
sequence of accesses in $T$ is completely free. For our
purposes we need that the accesses in $T$ respect the following
properties:
\begin{itemize}
\item When a node $v$ is visited by the cursor or involved in a
  rotation its depth is less than $3$, i.e., $d_T(v) < 3$.
\item After a rotation the cursor moves back to the root.
\end{itemize}
We refer to a sequence of nodes that respects these conditions
as a restricted sequence. Let us show that given any access sequence of
visited nodes in $T$ it is possible to simulate it in another
tree $T'$, so that the accesses in $T'$ are restricted. By simulation
we mean that the sequence visited by the cursor of $T$ is a
sub-sequence of the one visited by the cursor of $T'$. Naturally the
 simulation will be slower than the original sequence.
\begin{lemma}
\label{lem:simulateT}
  Any sequence of nodes visited by the cursor of $T$, consisting of
  $M$ moves and $R$ rotations, can be simulated by a restricted
  sequence of cursor moves on a tree $T'$ with $4M+3R$
  moves and $2M+R$ rotations.
\end{lemma}
\begin{proof}
\begin{figure}
  \centering
  \input{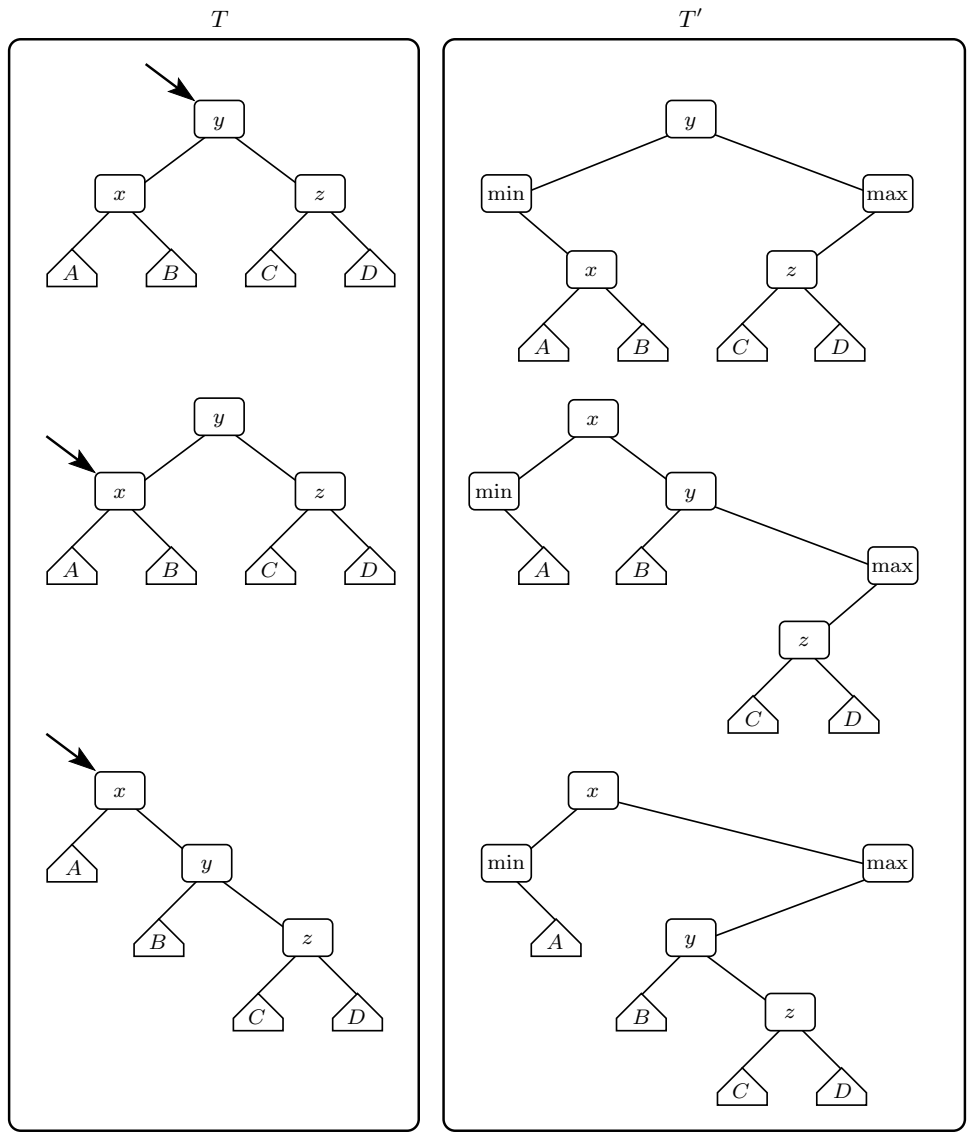}
  \caption{Simulation of $T$ (left) by $T'$ (right). Initial
    configuration on top. The middle configuration represents the
    result of moving the cursor from node $y$ to node $x$. The bottom
    configuration shows the result of performing a rotation on $x$,
    while in the middle configuration.}
  \label{fig:simul}
\end{figure}
Figure~\ref{fig:simul} illustrates the operations that we consider in
this Lemma. Besides the same keys as $T$, the tree $T'$ contains two extra keys,
$\min$ and $\max$, that are, respectively, smaller and larger than all
the other keys in $T$. The simulation works by keeping the node at the
cursor of $T$ in the root of $T'$.

We assume that the initial state of $T$ and $T'$ is almost alike. The
same key is stored a the root. The children of $T'$ are $\min$ and
$\max$. The left sub-tree of $T$ is stored in the right child of
$\min$ and the right sub-tree of $T$ is stored in the left child of
$\max$. Shown at the top in Figure~\ref{fig:simul}.

Whenever the cursor of $T$ moves down the corresponding node is moved
up to the root of $T'$, the process is a ZigZag operation on $T'$, or
ZagZig depending on which grand-child. Note that between the rotations
of the ZigZag operation the cursor must return to the root to comply
with the second condition of restricted sequences. These
movements are underlined in the example. Hence every downward
move in $T$ originates $4$ moves and $2$ rotations in $T'$. This
process transforms the trees in the top of Figure~\ref{fig:simul}, to
the trees in the middle. In this example the sequence of operations
is: \texttt{moveTo($\min$)}, \texttt{moveTo($x$)}, \texttt{rotate()},
\underline{\texttt{moveTo($y$)}, \texttt{moveTo($x$)}}
\texttt{rotate()}.

The cursor may also move upwards on the tree, hence
reversing the previous move. In this case the sequence of operations
would be: \texttt{moveTo($y$)},\texttt{rotate()},
\texttt{moveTo($x$)}, \texttt{moveTo($\min$)}, \texttt{rotate()},
\texttt{moveTo($y$)}. This requires $4$ moves and $2$
rotations. Therefore a move in $T$
originates $4$ moves and $2$ rotations in $T'$.

In this example the initial move was to $y$, because $y$ is the child
in $T'$ that corresponds to he parent in $T$. To determine this
property we could store, in $T'$ the depth of the nodes in $T$. Only
for the nodes in the leftmost and rightmost branches of $T'$,
otherwise updated values would be hard to maintain. However it is not
necessary to do so. Recall that $T'$ is simulating $T$ and $T$ is the
optimal tree, which does not need to consult the keys to know which
move to perform, it only needs to behave as a BST. Therefore $T'$ also
does not need extra information.

Whenever a rotation is performed in $T$ the structure of $T'$ must be
adapted accordingly. Recall that a rotation on node $v$ means
that $v$ is moved upwards. In Figure~\ref{fig:simul} the
transition from the middle to the bottom shows how a
rotation alters the structure of $T'$. In this case the sequence of
moves is: \texttt{moveTo($y$)}, \texttt{moveTo($\max$)},
\texttt{rotate()}, \texttt{moveTo($x$)}. Hence a rotation originates
$3$ moves and $1$ rotation in $T'$. The general procedure is to move
to $y$ because it is the child of $T'$ that corresponds to the parent
of $x$ in $T$, and move again in the same direction, in this case to
$\max$. This node is rotated upwards and the procedure finishes by
returning to the cursor back to the root.

Note that in Figure~\ref{fig:simul} the cursor of $T$ is drawn close
to the root. In general the structure of the upward path from the
cursor of $T$ to the root gets splitted into the leftmost and
rightmost branches of $T'$, but the update procedure is essentially as
explained.
\end{proof}

In the following results we continue to use $T$, instead of $T'$,
which simplifies the analysis. The tree $T'$ is used only in
Theorem~\ref{teo:optimal}.
\subsection{Amortized Costs of $S$ and $T$}
\label{sec:amortized-costs-s}
\hfill

Let us now bound the amortized time of
splaying a node of $S$.
\begin{lemma}
\label{lem:amortizedS}
Splaying a node $v$ takes at most $4 + 6 d_{T}(v)$, amortized time,
where $d_{T}(v)$ is the depth of the corresponding node in $T$.
\end{lemma}
\begin{proof}
The following derivation establishes this bound:
\begin{eqnarray}
\hat{c}_{i} & \leq & 1 + 3[r(t) - r(v)] \label{eq:aad1}\\
            & < & 1 + 3[1 - r(v)] \label{eq:aad2}\\
            & = & 4 - 3 r(v) \label{eq:aad3}\\
            & \leq & 4 - 3 \log(w(v)) \label{eq:aad4}\\
            & = & 4 - 3 \log(1/4^{d_{T}(v)}) \label{eq:aad5}\\
            & = & 4 + 6 d_{T}(v) \label{eq:aad6}
\end{eqnarray}
Inequality~(\ref{eq:aad1}) is the classic amortized access
Lemma~\ref{lem:amortizedAccess}, see
Section~\ref{sec:related-work}. Inequality~(\ref{eq:aad2}) follows
from~(\ref{eq:basic3}) and recalling that $r(t) =
\log(s(t))$. Inequality~(\ref{eq:aad4}) results from applying
logarithms and switching signs to~(\ref{eq:basic1}). The remaining
equations use the definition of weight $w$ and simplify the result.
\end{proof}

This Lemma shows that the amortized time to splay a node is slightly
more than $6$ times the time that it is necessary to access the
corresponding node in $T$. Let us focus on the amortized cost of
accessing nodes in $T$. Whenever $T$ accesses a node there is no real
cost for $S$. However if those accesses involve rotations in $T$ then
the value of $\Phi$ changes, which will constitute an amortized cost
that $S$ most account for.
\begin{lemma}
  \label{lem:powerDelta}
  Whenever a node $v$ of $T$, with $d_T(v) < 3$, gets rotated the
  bound $\Delta \Phi \leq 11 + \log (11)$ holds, after, at most, $3$ nodes of $S$
  are splayed. Each splayed node $v^*$ has $d_T(v^*) \leq d_T(v)$.
\end{lemma}
\begin{proof}
Recall that in a rotation the node $v$ moves upwards on the tree.
We consider the cases when $d_T(v) = 1$ and $d_T(v) = 2$.
\begin{figure}[bt]
\begin{center}
\input{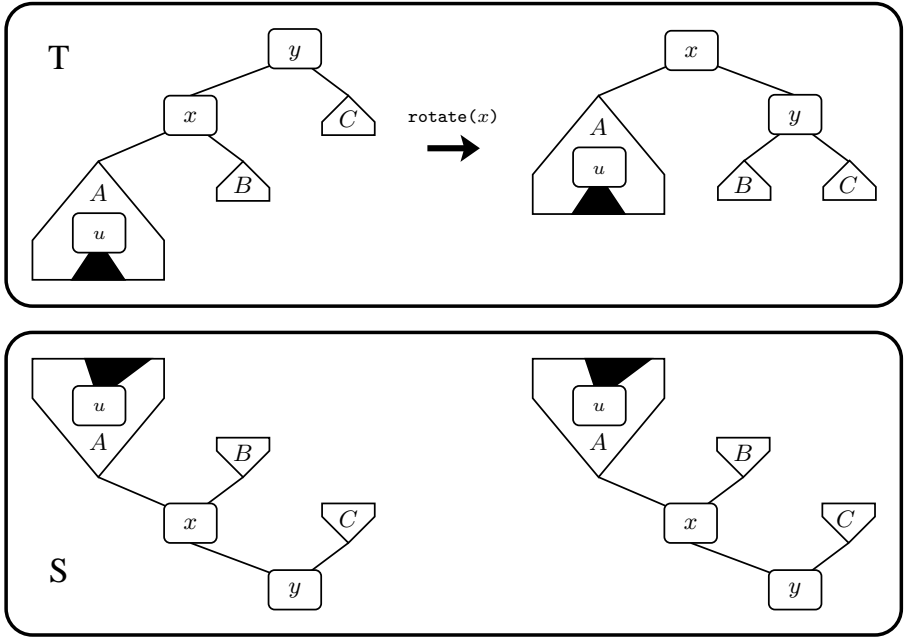}
\caption{Variation of $\Phi$ with rotation on T.}
\label{F:rotP}
\end{center}
\end{figure}

The proof is illustrated in Figure~\ref{F:rotP}, which shows
the situation when $d_T(v) = 1$, but from which we can infer
properties that apply to all cases.

 The optimal tree $T$ is represented on top and the splay tree $S$ in
 the bottom, upside down. To simplify let us assume that these trees
 are only slightly different. Meaning that they share some
 structure. The nodes in sub-trees $A$, $B$ and $C$ are the same, but
 their shape is not necessarily equal. For example the
 descendants of node $u$ are not, necessarily, the same.

A fundamental observation for this proof is that almost all the nodes,
in the trees, contribute $0$ to the value $\Delta \Phi$. The only
nodes that effectively contribute to $\Delta \Phi$ belong to, at
least, one the following categories:
\begin{enumerate}
\item Nodes that contain descendants from more than one of the sets
  $A$, $B$, $C$, either in $S$ or $T$ or both. Examples of these nodes
  are the ones containing $x$ and $y$. Moreover some nodes $a \in A$,
  $b \in B$ and $c \in C$ may appear in $S$ with this property.
\item Nodes for which the set of descendants is altered by the
  rotation in $T$. Only nodes $x$ and $y$.
\end{enumerate}

Hence we are claiming that the nodes in $A$, $B$ and $C$ do not contribute to
  $\Delta \Phi$.  Consider a node $u$, contained in the sub-tree
  $A$. Let $w_T(u)$ be the weight of $u$ before the rotation and
  $w_T'(u)$ be the weight after the rotation. Likewise we consider the
  $s_T$, $s_T'$, $r_T$, $r_T'$ values and the values $s$, $s'$, $r$,
  $r'$, over $S$.

  Let us account for the contribution of $u$ to $\Delta \Phi$, i.e.,
  $r'(u)-r(u)+r_T(u)-r'_T(u)=
  \log([s'(u)/s(u)]\times[s_T(u)/s'_T(u)])$.  Notice that $w'(u)/w(u)=
  4$, because the depth of $u$ decreases in $T$. Likewise
  $s'(u)/s(u)=4$, because the same happens for all the elements in
  $A$. On the other hand $s_T(u)/s'_T(u)= 1/4$ because the order of
  the factors is reversed. Therefore the previous value is $\log (4/4) =
  0$. A similar reasoning holds for the elements in $C$. The nodes in
  $B$ do not suffer potential variations and therefore contribute $0$
  to the overall variation.

  In this scenario the only nodes that contribute to $\Delta \Phi$ are
  the ones that contain descendants from more than one of the $A$,
  $B$, $C$ sets. In our simplification only $x$ and $y$. However in
  the general case the nodes in sub-trees $A$, $B$ and $C$ have
  different relations in $S$ and in $T$. Hence we need to be more
  precise as to what these nodes should be. The nodes in $A$ are the
  descendants of $x$ that are smaller than $x$. The nodes in $B$ are
  the descendants of both $x$ and $y$. The nodes in $y$ are the
  descendants of $y$ that are larger than $y$.

  To force these sets of nodes to be consistent between $S$ and
  $T$ we splay node $x$ and node $y$ in $S$, in this order. Now $A$ and
  $C$ contain the same nodes in $S$ and $T$, but it may happen that
  some ZigZig operation pulls a node $b$ upwards from $B$, so that it
  splits $B$ and is in between the nodes $x$
  and $y$. In this case we must also count the contribution from $b$.

  Using an argument similar to the one above we can guarantee that
  $s'(b)/s(b) \leq 4$ and $s_T(b)/s'_T(b) = 1$ and therefore
  $r'(b)-r(b)+r_T(b)-r'_T(b) \leq 2+0 = 2$.

  For the node containing $x$ we also have that $s'(x)/s(x) \leq 4$,
  but $s_T(x)/s'_T(x)$ is trickier, because $s'_T(x)$ now includes the
  extra terms related to $w_T'(y)$ and $s_T'(C)$, these
  values are non-negative and therefore
  decrease the fraction $s_T(x)/s'_T(x)$. Therefore the bound of $4$
  holds true. Hence for the node
  containing $x$ we count another $4$ units, as
  $r'(x)-r(x)+r_T(x)-r'_T(x) \leq 2+2 = 4$.

  For the node containing $y$ we also have that $s'(y)/s(y)
 \leq 4$ but $s_T(y)/s'_T(y)$ is again trickier. The following
 derivation obtains a bound.
  \begin{eqnarray}
    s_T(y)/s'_T(y) & = &
    \frac{w_T(y)+s_T(C)+s_T(B)+w_T(x)+s_T(A)}{s'_T(y)} \label{eq:6} \\
& = & \frac{w_T(y)+s_T(C)+s_T(B)}{s'_T(y)} +
\frac{w_T(x)+s_T(A)}{s'_T(y)}  \label{eq:7}\\
& \leq & 4 + \frac{w_T(x)+s_T(A)}{s'_T(y)} \label{eq:8} \\
& \leq & 4 + \frac{w_T(x)+s_T(A)}{1/4} \label{eq:9} \\
& = & 4 + \frac{(1/4)+s_T(A)}{1/4} \label{eq:10} \\
& \leq & 4 + \frac{(1/4)+(1/8)}{1/4} \label{eq:11} \\
& \leq & 4 + 1.5 \label{eq:12}     
  \end{eqnarray}

  Equations~(\ref{eq:6}) and~(\ref{eq:7}) are simple
  manipulations. Inequality~(\ref{eq:8}) is our general bound of $4$,
  for the fraction that does not change the
  descendants. Inequality~(\ref{eq:basic1}) yields $1/4 = w'_T(y) \leq
  s'_T(y)$, which we use to establish
  Inequality~(\ref{eq:9}). Inequality~(\ref{eq:11}) follows from
  Inequality~(\ref{eq:basic2}) that yields $s_T(A) \leq 2 \times 1/16 =
  1/8$.  Therefore the total for $y$ is $r'(y)-r(y)+r_T(y)-r'_T(y)
  \leq 2 + \log(5.5) = 1 + \log(11)$.

  Summing up, in this case we splayed two nodes and have to account
  for $b$, $x$ and $y$, therefore
  $\Delta \Phi \leq 2 + 4 + (1 + \log(11)) = 7 + \log(11)$.
  
  Let us consider the case $d_T(v) = 2$. This means that there is a
  node $v_r$ which is an ancestor of $y$, in $T$. Let us assume
  that $v_r$ is to the right of $C$. The case where the node is to the
  left of $A$ is easier. It may happen that
  $v_r$ is inside $C$ in $S$, in that case we splay $v_r$. This
  splay operation may in turn split $C$ by pulling up a node $c$. In
  total we splayed 3 nodes and have $5$ nodes that may contribute to
  $\Delta \Phi$, the nodes $b$ and $c$, that got pulled up, and the
  nodes $x$, $y$, $v_r$.

  For node $c$ we have a bound of $4$, using the argument above, which
  holds  because the descendants of $c$ in $T$ do not
  change. Hence $r'(c)-r(c)+r_T(c)-r'_T(c) \leq 2+2 = 4$.

  Note that $v_r = t$ is the root of both $S$ and $T$. Therefore we
  can can bound its variation by $0$, instead of our pessimistic $4$
  value. For this node we have $r'(t) = r'_T(t)$ and $r(t) = r_T(t)$,
  by Equation~(\ref{eq:basic5}). Hence in this case the total bound is
  $\Delta \Phi \leq 2 + 4 + 4 + (1 + \log(11)) + 0 = 11 + \log (11)$

  Notice that all this splaying may change the ancestor relation
  between $x$ and $y$ in $S$.  This may happen when there is no node
  $b$ and we are pulling a node $v_\ell$, to the left of $A$. This
  relation in $S$ does not change the argument above, as both $x$ and
  $y$ are being pessimistically bounded in $S$. The crucial
  consequence of the extra splays is that the nodes in $A$, $B$ and
  $C$ are properly encapsulated, and the number of nodes whose
  descendants belong to more than one of these sets is limited.
\end{proof}

We can now prove our optimality result. 
\begin{theorem}
\label{teo:optimal}
Consider a sequence of $m$ queries, to a splay tree $S$ with $n+2$ keys,
for which, a similar, optimal tree $T$ uses $R$ rotations and $M$ cursor
movements. Provided that splay trees have regular access, then the
total number of operations performed by $S$ is $O(n + R + M)$.
\end{theorem}
\begin{proof}
The proof is given by the following deduction:
\begin{eqnarray}
\sum_{1 \leq i \leq m} c_i - k (n+2) & \leq & k \left( \sum_{1\leq i \leq
    m+e} \ell + c'_i \right) - k (n+2)\label{Eq:R0}  \nonumber \\
 & = & k \left( (m+e) \ell - (n+2) + \sum_{1\leq i \leq
    m+e} c'_i \right) \label{Eq:R1}\\
& \leq & k \left((m+e) \ell + \Phi(D_m) + \sum_{1\leq i \leq m+e} c'_i \right) \label{Eq:R2} \\
& = & k \left((m+e) \ell + \Phi(D_m) - \Phi(D_0) + \sum_{1\leq i \leq m+e} c'_i \right) \label{Eq:R3} \\
& = & k \left((m+e) \ell + \sum_{1\leq i \leq m+e} \hat{c}'_i
\right) \label{Eq:R4} \\
& \leq & k \left[(m+e) \ell + (11 + \log 11) R' + (1 + 3) (4 + 6 M')
\right] \label{Eq:R5} \\
& \leq & k \left[(m+3R') \ell + (11+\log 11) R' + (1 + 3) (4 + 6 M')
\right] \label{Eq:R6} \\
& = & k \left[m \ell + 16 + (11 + 3 \ell + \log 11) R' + 24 M'
\right] \label{Eq:R7} \\
& = & k \left[m \ell + 16 + (11 + 3 \ell + \log 11) (2M + R) + 24(4M+3R) \right] \label{Eq:R8} \\
& = & k \left[m \ell + 16 + (118 + 6 \ell + 2 \log 11) M + (83 + 3 \ell
  + \log 11) R  \right] \label{Eq:R9}
\end{eqnarray}
The above deduction uses the restricted sequence over $T'$, instead
of the original optimal sequence over $T$. We point out which of the
results in the previous lemmas apply to $T'$. The values $R'$ and $M'$ 
refer to the number of rotations and cursor movements in $T'$.

We aim to bound the value on the left.
 The first inequality is our regular access conjecture. In here we are
 using explicit constants, $k$ and $\ell$, instead of the notation
 $O$, to determine the hidden factors. Equation~(\ref{Eq:R1})
 rearranges the sums. Equation~(\ref{Eq:R2}) follows from
 Lemma~\ref{lem:bound-phi}. In Equation~(\ref{Eq:R3}) we add the value
 $\Phi(D_0)=0$, because we are assuming that the initial structure of
 $S$ and $T'$ is the same.  Equation~(\ref{Eq:R4}) follows from the
 telescopic nature of the definition of amortized costs.

Inequality~(\ref{Eq:R5}) follows from Lemmas~\ref{lem:amortizedS}
and~\ref{lem:powerDelta}. From Lemma~\ref{lem:amortizedS} we conclude that the sum of the
$\hat{c}'_i$ that corresponds to splay operations that answer queries
is at most $(4 + 6 M')$, because $T'$ must also move its cursor to
those nodes and between any such nodes the cursor returns to the
root. Therefore for $T'$ to answer a query it most perform at least
$d_T(v)$ cursor movements, where $v$ is the node containing the key
that corresponds to the query.

To bound the sum of the $\hat{c}'_i$ that corresponds to the costs to
update $S$, due to rotations in $T'$, we use
Lemma~\ref{lem:powerDelta}. Note that we can apply this Lemma because
of the first property of restricted sequences. The fixed cost yields
the term $(11 + \log 11) R'$. The amortized cost of splaying, at most,
$3$ nodes of $S$ is bounded by the term $3(4 + 6 M')$. This bound
holds because the second property of
restricted sequences implies that $T'$ must perform at least $d_T(v)$
cursor movements before rotating node $v$. This
property states that after each rotation
the cursor of $T'$ returns to the root, see
Section~\ref{sec:restricting-t}.

In Equation~(\ref{Eq:R6}) we use the fact that $e \leq 3R'$, i.e., we
use at most $3$ extra splays for each rotation in
$T'$. Equation~(\ref{Eq:R7}) results by rearranging terms. In
Equation~(\ref{Eq:R8}) we use the bounds from
Lemma~\ref{lem:simulateT}. Equation~(\ref{Eq:R9}) obtains the final
bound, by rearranging terms.
\end{proof}

Let us now survey related work, so that we can discuss our result in
context.

\section{Related Work}
\label{sec:related-work}
\hfill

This section describes splay trees and some previous
results. These trees were proposed by Sleator and
Tarjan~\cite{Sleator:1985:SBS:3828.3835} in 1985. An extensive up
to date survey on dynamic optimally was given by
Iacono~\cite{Survey}.

There are several restricted optimally results for BSTs. Assuming that
the optimal BST is static than the best performance is the
entropy $O(\sum_{i=1}^m f_i \log(m/f_i))$, where $m$ is the
total number of queries and $f_i$ is the number of queries for the
$i$-th key. Knuth presented the first algorithm to determine such an
optimal tree~\cite{Knuth:1998:ACP:280635}. This dynamic programming
algorithm requires $O(n^2)$ time to determine the optimal tree. Kurt
Mehlhorn obtained a faster, $O(n)$ time algorithm, which approximates
the optimal tree~\cite{Mehlhorn}. In the context of this paper this
result would be referred to as statically optimal, because we are
hiding factors into the $O$ notation. Splay trees also obtain this
kind of optimality and moreover do not need to know the $f_i$
values. Static optimality becomes the best upper bound when the keys
in the query sequence are independent of each other.

We are only interested in approximating dynamic optimality, i.e.,
using $O$ notation and being a factor away from the optimal value,
because it is expected that the exact problem is
NP-Complete~\cite{BST_SODA2009}. Even in this case the only BST that
achieves the dynamic optimality bound, assumes free rotations and
takes exponential time to select
operations~\cite{DBLP:journals/algorithmica/BlumCK03}.

The good performance properties of splay trees are due to the access
Lemma, which was establish in the paper on splay
trees~\cite{Sleator:1985:SBS:3828.3835}. Here we repeat the original
argument and highlight the fact that it follows from Jenssen's
Inequality~\cite{Jenssen}.
\begin{lemma}[Amortized Access]
\label{lem:amortizedAccess}
Splaying a node $v$ to the root $t$ takes at most $1+3[r(t)-r(v)]$,
amortized time.
\end{lemma}
\begin{proof}
  Recall Figures~\ref{F:zig},~\ref{F:zigzig} and~\ref{F:zigzag}.  When
  $v$ is at the root, the bound is trivial. Hence let us focus on the
  last operation that is used to access $v$, which contains $x$ in our
  figures.
\begin{description}
\item[Zig] In this case $x$ is a child of $y$, which is at the
  root. The amortized cost is computed as follows:
\begin{eqnarray}
\hat{c}_{Zig} &=& 1 + r'(x) + r'(y) - r(x) - r(y) \label{eq:z1}\\
 & \leq & 1 + r'(x) - r(x) \label{eq:z2} \\
 & \leq & 1 + 3[r'(x) - r(x)]
\end{eqnarray}
The only nodes that change rank are $x$ and $y$. The
Inequality~(\ref{eq:z1}) follows from the fact that $r'(y) \leq
r(y)$. On the other hand Inequality~(\ref{eq:z2}) is true because
$r(x) \leq r'(x)$.
\item[ZigZig] In this case $x$, $y$ and $z$ change rank
\begin{eqnarray}
\hat{c}_{ZigZig} & = & 2 + r'(x) + r'(y) + r'(z) - r(x) - r(y) - r(z) \label{eq:zz1}\\
& = & 2  + r'(y) + r'(z) - r(x) - r(y) \label{eq:zz2} \\
& \leq & 2  + r'(x) + r'(z) - 2 r(x) \label{eq:zz3} \\
& \leq & 3[r'(x)-r(x)] \label{eq:zz4} 
\end{eqnarray}
Equation~(\ref{eq:zz2}) follows from the fact that $r'(x) =
r(z)$. Inequality~(\ref{eq:zz3}) is true because $r'(y) \leq r'(x)$ and
$- r(y) \leq - r(x)$. Inequality~(\ref{eq:zz4}) reduces to proving that $
(r(x) + r'(z))/2 \leq r'(x) - 1$, this relation is known as Jensen's
inequality and it holds because the $\log$ function is concave. More
explicitly the relation is the following:
\begin{equation*}
\frac{\log(s(x)) + \log(s'(z))}{2} \leq \log\left(\frac{s(x)+s'(z)}{2}\right)
\end{equation*}
\item[ZigZag] In this case $x$, $y$ and $z$ change rank.
\begin{eqnarray}
\hat{c}_{ZigZag} & = & 2 + r'(x) + r'(y) + r'(z) - r(x) - r(y) - r(z) \label{eq:za1}\\
& = & 2  + r'(y) + r'(z) - r(x) - r(y) \label{eq:za2} \\
& \leq & 2  + r'(y) + r'(z) - 2 r(x) \label{eq:za3} \\
& \leq & 2[r'(x)-r(x)] \label{eq:za4} \\
& \leq & 3[r'(x)-r(x)] \label{eq:za5}
\end{eqnarray}
Equation~(\ref{eq:za2}) follows from the fact that $r'(x) =
r(z)$. Inequality~(\ref{eq:za3}) is true because $- r(y) \leq
-r(x)$. Inequality~(\ref{eq:za4}) is also Jensen's inequality, in this
case it reduces to $(r'(y)+ r'(z))/2 \leq r'(x) - 1$. The
Inequality~(\ref{eq:za5}) is true because $r(x) \leq r'(x)$.
\end{description}
Notice that in any access to $S$ there is at most $1$ Zig operation
and several ZigZig and ZigZag. This justifies why it is important to
omit the term $1$ in the bound for the ZigZig and ZigZag
operations. The overall bound is obtained by summing the bounds of the
respective operations. This sum telescopes to the expression in the
Lemma.
\end{proof}

A recent detailed study of this Lemma is
available~\cite{DBLP:conf/esa/ChalermsookG0MS15}. An important
consequence of the splay operation is that most of the nodes in the
splayed branch are moved upwards, i.e., their depth gets reduced,
essentially in half. A
detailed study on depth reduction was given by
Subramanian~\cite{DBLP:journals/jal/Subramanian96}.

Several good performance theorem follow from this Lemma:
\begin{description}
\item[Static Optimality,] obtaining the entropy bound we mentioned in
  the beginning of the Section. 
\item[Static Finger,] the performance of splay trees can
  be made to made dependent on the position of a given element of the
  key set $X$. Let $f$ be the position of a certain element $x$ in the
  ordered set $X$ and $|x_i-f|$ the distance between that element and
  the element $x_i$, the $i$-th element in the query sequence. Then
  processing a sequence of queries with a splay tree takes $O(n \log n
  + m + \sum_{i=1}^m\log(|x_i-f|+1))$. In fact splay trees have an even better
  better performance because the chosen element at position $f$ can be
  made
  dynamic~\cite{doi:10.1137/S0097539797326988,doi:10.1137/S009753979732699X},
  but this result requires a longer analysis.
\item[Working Set,] meaning that recently accessed are quicker to
  access again. Let $t(x)$ be the number of different items accessed
  before accessing $x$ since the last time $x$ was accessed, or since
  the beginning of the sequence if $x$ was never accessed before. The
  the total time to process a sequence of queries is $O(n \log n + m +
  \sum_{i=1}^m \log(t(i)+1))$. An important consequence of this bound
  is that if the keys are assigned arbitrarily (but consistently) to
  unordered data, splay trees behave as dynamically
  optimal~\cite{DBLP:journals/algorithmica/Iacono05}.
\item[Unified bound,] combining the best performance of the 3 previous
  bounds. Dedicated structures where designed to achieve an even
  better bound which uses the dynamic finger
  bound~\cite{DBLP:journals/tcs/BadoiuCDI07}, instead of the static
  version. Generalized versions of this bound, with tighter values for
  structured sequences where also
  proposed~\cite{DBLP:journals/corr/abs-1302-6914}.
\end{description}
Another important performance bound is the sequential access, or
scanning, bound, which states that accessing the $n$ keys sequentially
takes only $O(n)$ time. This was initially established by
Tarjan~\cite{TarjanSeq} with a factor of $9$, the current best factor
of 4.5~\cite{Elmasry2004459} was proven by Elmasry.

Besides dynamic optimality  Sleator and
Tarjan~\cite{Sleator:1985:SBS:3828.3835} proposed the traversal
conjecture, which remains unproven. It is similar to the sequential
access but the key order is taken to be the preorder of another
BST.

Other open conjectures on splay trees include the Deque conjecture,
which claims that splay trees can be used to implement a deque, with
$O(1)$ amortized time per
operation~\cite{dequeCom,DBLP:journals/corr/abs-0707-2160}. The split
conjecture claims that deleting all the nodes of a splay tree takes
$O(n)$ time, in whatever order~\cite{lucas1992competitiveness}.

Studying binary search trees from a geometric point of view has
yielded several important results~\cite{BST_SODA2009}. It presented an
online algorithm, which may yield dynamically optimal BSTs. This
algorithm is referred to as greedy, and it was originally proposed as
an offline
algorithm~\cite{lucas1988canonical,munro2000competitiveness}. The
first $O(\log n)$ performance bound, of this algorithm, was
established by Fox~\cite{DBLP:conf/wads/Fox11}.

The same authors proposed Tango trees which are $O(\log \log n)$
competitive to the optimal BST~\cite{Tango_SICOMP}. This was the first
structure to obtain a proven non-trivial competitive ratio. The worst
case performance of these trees was further
improved~\cite{DBLP:conf/wads/DerryberryS09} to $O(\log n)$ per
operation. The $O(\log \log n)$ ratio was also proved for the
chain-splay variation of splaying~\cite{Georgakopoulos200837}.

Another alternative to try and obtain, proven, dynamic optimality
consists in combining BSTs. In~\cite{ComboBST_ICALP2013} the authors
show that given any constant number of online BST algorithms (subject
to certain technical restrictions), there is an
online BST algorithm that performs asymptotically as well on any
sequence as the best input BST algorithm. 

Iacono~\cite{Survey} recently proposed another approach, the Weighted
Majority Algorithm. This approach is proven to yield a dynamically
optimal BST, provided any such data structure exists.

A considerable amount of research as also gone into the study of lower
bounds, i.e., formulating expressions that are smaller than $R+M$,
i.e., the amount of operations required by $T$. The first lower bounds
where established by
Wilber~\cite{DBLP:journals/siamcomp/Wilber89}. These bounds where
further improved by the geometric view of BSTs~\cite{BST_SODA2009}. In
this view accesses to nodes in a tree are drawn as points in the
plane. One coordinate stores the ordered keys and the other coordinate
represents time. We can then consider rectangles using these points as
corners. Two rectangles rectangles are independent if their corners
are outside their interception or in the boundary of the
interception. An access sequence in a BST must respect any independent
set of rectangles. Hence the maximum independent set of rectangles
yields a lower bound for the number of visited nodes in a BST.

The alternation lower and the funnel bounds can also be explained with
the geometric view of BSTs. Moreover this latter bound is also related
to the number of turns that it is necessary to move a key to root,
with a procedure proposed by Allen and
Munro~\cite{DBLP:journals/jacm/AllenM78}.

\section{Conclusions and Further Work} 
\label{sec:concl-furth-work}
\hfill

In this paper we studied the dynamic optimality conjecture of splay
trees, which was proposed by the authors of these
trees~\cite{Sleator:1985:SBS:3828.3835} and has stood for more than 30
years. During this time a vast amount of results and approaches have
been proposed for this problem and related conjectures. Note that a
proof of dynamic optimality would establish other conjectures. 

All this research has resulted in new data structures, several of
which where proven to be $(\log \log n)$ competitive. The results we
present are a step forward. We reduced dynamic
optimality to the regular access conjecture. This conjecture seems
very plausible. On the other hand, if it is false that would also be a
remarkable property.

In the future we plan to research this conjecture. An exhaustive
literature review may provide the necessary tools to solve it. We also
plan to investigate if our potential function can be applied to the
analysis of other dynamically optimal candidates, namely the greedy
algorithm, combining BSTs or the weighted majority algorithm.

\bibliographystyle{siam}
\bibliography{manuscript.bib}

\end{document}